\documentclass[a4paper,UKenglish,cleveref, autoref, thm-restate]{lipics-v2021}

\usepackage[ruled,vlined,linesnumbered,commentsnumbered]{algorithm2e}
\usepackage{xcolor}
\usepackage{hyperref}

\nolinenumbers

\title             {A tight example for approximation ratio 5 for covering small cuts by the primal-dual method}
\titlerunning{A tight example for 5-approximation for covering small cuts by primal-dual}

\author{Zeev Nutov}{The Open University of Israel}{nutov@openu.ac.il}
{https://orcid.org/0000-0002-6629-3243}{}
\authorrunning{Zeev Nutov}
\Copyright{Zeev Nutov}


\begin{document}

\maketitle



\def\sem    {\setminus}
\def\subs   {\subseteq}
\def\empt   {\emptyset}

\def\f   {\frac}

\def\de   {\delta}

\def\AA  {{\cal A}}
\def\CC  {{\cal C}}
\def\FF   {{\cal F}}


\keywords{primal dual approximation algorithm, small cuts cover problem, tight example}

\begin{abstract}
In the {\sc Small Cuts Cover} problem 
we seek to cover by a min-cost edge-set the set family 
of cuts of size/capacity $<k$ of a graph. 
Recently, Simmons \cite{S} showed that the primal-dual algorithm 
of Williamson, Goemans, Mihail, and Vazirani \cite{WGMV}
achieves approximation ratio $5$ for this problem, and asked whether this bound is tight.
We will answer this question positively, by providing an example in which the ratio between 
the solution produced by the primal-dual algorithm and the optimum is arbitrarily close to $5$. 
\end{abstract}

\section{Introduction} \label{s:intro}

Let $G=(V,E)$ be a graph. 
We call any proper node subset $\empt \ne S \subset V$ of $V$ a {\bf cut}.
Let $\de(S)$ denote the set 
of edges of $G$ with exactly one end in $S$ and let $d(S)=|\de(S)|$ be their number;
for singletons we write $\de(a)$ and $d(a)$ instead of $\de(\{a\})$ and $d(\{a\})$, respectively.
An edge set $J$ {\bf covers} a cut $S$ if there is an edge in $J$ with exactly one end in $S$.
Given an integer $k$ we say that $S$ is a {\bf small cut} of $G$ if $d(S)<k$.
We consider the following problem.

\begin{center}
\fbox{\begin{minipage}{0.98\textwidth} \noindent
\underline{\sc Small Cuts Cover} \\ 
{\em Input:} \ \ A graph $G=(V,E)$, 
a set $L$ of links (edges) on $V$ with costs $\{c_e:e \in L\}$, and an \\
\hspace*{1.2cm} integer $k \ge 1$. \\ 
{\em Output:} A min-cost link set $J \subs L$ that covers the set-family $\FF=\{\empt \ne S \subset V:d(S) <k\}$ \\
\hspace*{1.2cm} of small cuts of $G$.
\end{minipage}} \end{center}

Following the breakthrough achieved in \cite{BCGI}, 
where it was shown that a constant approximation ratio can be achieved
using the primal-dual algorithm of 
Williamson, Goemans, Mihail, and Vazirani \cite{WGMV},
the performance of the primal-dual algorithm for this problem 
was studied in several other papers, c.f. \cite{N-prel, B, N-MFCS,SBC}.
Recently, Simmons \cite{S} showed that the primal-dual algorithm 
achieves approximation ratio $5$ for this problem, and asked whether this bound is tight.
We will answer this question positively, by providing an example in which the ratio between 
the solution produced by the primal-dual algorithm and ${\sf opt}$ (the optimal solution cost)
is arbitrarily close to $5$. 

\begin{theorem} \label{t:main}
For any integers $q,p,k$ such that 
$q \ge 1$, $p \ge 2$, and $k \ge 2pq+1$
there exists an instance of the {\sc Small Cuts Cover} problem
on $4p+3$ nodes such that the WGMV  \cite{WGMV} primal-dual algorithm 
computes a solution of cost $\f{5p}{p+2} \cdot {\sf opt}$.
\end{theorem}

\section{The construction} \label{s:main}

We start by describing the WGMV algorithm for covering an arbitrary set family $\FF$. 
An inclusion-minimal set in $\FF$ is called an {$\FF$-\bf core}, 
or just a {\bf core}, if $\FF$ is clear from the context.  
Consider the standard LP-relaxation {\bf (P)} for the problem of covering the family $\FF$ 
and its dual program {\bf (D)}:
\[ \displaystyle
\begin{array} {lllllll} 
& \hphantom{\bf (P)} & \min              & \ \displaystyle \sum_{e \in E} c_e x_e                                                                 & 
    \hphantom{\bf (P)} & \max            & \ \displaystyle \sum_{S \in \FF} y_S  \\
& \mbox{\bf (P)}         & \ \mbox{s.t.} & \displaystyle \sum_{e \in \delta(S)} x_e  \geq 1   \ \ \ \forall S \in \FF \ \ \ \ \ \ \  &
    \mbox{\bf (D)}        & \ \mbox{s.t.} & \displaystyle \sum_{\delta(S) \ni e} y_S \leq c_e \ \ \ \forall e \in E \\
& \hphantom{\bf (P)} &                      & \ \ x_e \geq 0 \ \ \ \ \ \ \ \ \ \forall e \in E                                                                  &
    \hphantom{\bf (P)} &                      & \ \ y_S \geq 0 \ \ \ \ \ \ \ \ \ \ \forall S \in \FF 
\end{array}
\]

Given a solution $y$ to {\bf (D)}, an edge $e \in E$ is {\bf tight} if 
the inequality of $e$ in {\bf (D)} holds with equality.
The algorithm has two phases.

{\bf Phase~1} starts with $J=\emptyset$ and applies a sequence of iterations.
At the beginning of an iteration, we compute the cores of the family of sets in $\FF$ not yet covered by $J$.
Then we raise the dual variables corresponding to these cores
uniformly until some edge $e \in E \sem J$ becomes tight, and add all tight edges to $J$.
Phase~1 terminates when $J$ covers $\FF$.

{\bf Phase~2} is a ``reverse delete'' phase, in which we process edges in the reverse order
that they were added, and delete an edge $e$ from $J$ if $J \sem \{e\}$ still covers $\FF$.
At the end of the algorithm, $J$ is output.

\medskip

Consider the {\sc Small Cuts Cover} instance $G,c,L$ in Fig.~\ref{f:g}(a), 
where $q \ge 1$ and 
\begin{equation} \label{e:p1}
k \ge 2q +1
\end{equation}
The set of links is a disjoint union of two sets: 
\[
\mbox{
the set of {\bf red links}   $L_{\sf red}=\{tx,ay,yr\}$ and 
the set of {\bf blue links} $L_{\sf blue}=\{tb,rz\}$.}
\]
The cost of a link is the number of non-white nodes it connects, so 
$c(ya)=c(tb)=1$ and the weight of any other link is $2$. 

\begin{figure} \centering \includegraphics[scale=0.7]{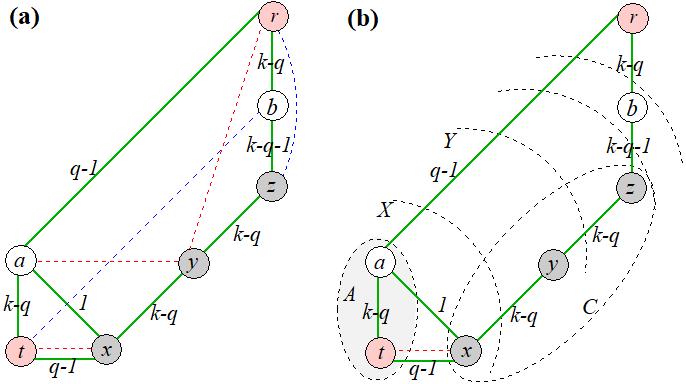}
\caption{
\begin{enumerate}[(a)]
\item
The instance of {\sc Small cuts Cover} $G,c,L$. 
The green edges are those of the input graph, and the number on a green edge 
is the multiplicity/capacity of the edge. The red and blue edges are links.
The weight of a link is the number of non-white nodes it connects. 
\item
The $\FF$-cores are $\{t\}$, $\{r\}$ and $C$, and 
the cuts $A,X,Y$ in $\FF^r \sem \FF^r_C$.
The figure also shows two additional cuts in $\FF^r_C$,
but these cuts are not addressed in Lemma~\ref{l:cores}. 
\end{enumerate}
}
\label{f:g} \end{figure}

Let $\FF$ be the family of small cuts of $G$ and let $\CC$ be the family of $\FF$-cores.
Note that the singleton sets $\{t\}$ and $\{r\}$ are small cuts of $G$ 
(since $d(t)=d(r)=k-1$)
and that the following sets are also small cuts of $G$, see Fig.~\ref{f:g}(b):
\begin{eqnarray*}
A =\{t,a\}          & & d(A)=2q-1\le k-2 \\
X=A \cup \{x\}  & & d(X)=k-1  \\ 
Y=X \cup \{y\}  & & d(Y)=k-1     \\
C= \{x,y,z\}        & & d(C)= k-1
\end{eqnarray*}
Since $\FF$ is symmetric ($S \in \FF$ implies $V \sem S \in \FF$),
it is sufficient to consider small cuts of $G$ not containing $r$. 
Let us define two families of small cuts:
\begin{itemize}
\item
$\FF^r=\{S \in \FF: r \notin S\}$ is the family of small cuts of $G$ that do not contain $r$.
\item
$\FF^r_C=\{S \in \FF: r \notin S, C \subs S\}$ is the 
family of small cuts of $G$ that contain $C$ but not $r$. 
\end{itemize}

\begin{lemma} \label{l:cores}
The family of $\FF$-cores is $\CC=\{\{r\},\{t\},C\}$ and $\FF^r \sem \FF^r_C = \{\{t\},A,X,Y\}$.
\end{lemma}
\begin{proof}
We prove the first part. As was already observed, $\{r\},\{t\}$ and $C$ are small cuts.
Thus the singleton sets $\{r\}$ and $\{t\}$ are $\FF$-cores and $C \in \FF$.
To show that $C$ is an inclusion-minimal small cut we will just verify 
that $d(S) \ge k$ for any $\empt \ne S \subset C$:
\[ 
d(x)=k \ \ \ \ \  d(y)=2k-2q \ge k+1 \ \ \ \ \  d(z)=2k-2q-1 \ge k 
\]
\[
d(\{x,y\})=k \ \ \ \ \ d(\{y,z\})=2k-2q-1 \ge k \ \ \ \ \ d(\{x,z\})=d(x)+d(z)=k+d(z)
\]
The inequalities are by (\ref{e:p1}). 
Since the $\FF$-cores are pairwise disjoint, the only additional candidates to be $\FF$-cores 
are the sets $\{a\}$, $\{b\}$ and $\{a,b\}$.  
But 
\[
d(a)=k \ \ \ \ \  d(b)= 2k-2q-1 \ge k \ \ \ \ \ d(\{a,b\})=d(a)+d(b)=k+d(b) 
\]
concluding the proof of the first part that $\CC=\{\{r\},\{t\},C\}$.

We prove the second part.
We already know that $\{\{t\},A,X,Y\} \subseteq \FF^r \sem \FF^r_C$. 
Suppose to the contrary that $\FF^r \sem \FF^r_C$ has some other cut $S$. 
Since $S$ contains a core distinct from $C$, $t \in S$. 
Since $S \ne \{t\}$, we must have $a \in S$, as otherwise 
$\de(S)$ contains at least $k-1$ edges incident to $a$ and some other edge,
implying the contradiction $d(S) \ge k$.
Thus $S$ properly contains $A$ and $\de(S)$ contains the $q-1$ $ar$ edges. 
If $b \in S$ then $\de(S)$ contains additional $k-q$ $br$ edges,
and thus $S$ must contain $z,y,x$, as otherwise 
$d(S) \ge (q-1)+(k-q)+(k-q)>k$.
This gives the contradiction $C \subs S$, so $b \notin S$.
Similarly, if $z \in S$ then $\de(S)$ contains $k-q-1$ $zb$ edges,  
and thus $S$ must contain $y,x$, as otherwise $d(S) \ge (q-1)+(k-q-1)+(k-q)> k$.
This again gives the contradiction $C \subs S$, so $z \notin S$.
Since $S \ne X$ and $S \ne Y$, the only case is $S=A \cup \{y\}$, 
but this is not possible since then $d(S)=d(A \cup \{y\}) >d(y)=2k-2q>k$.
\end{proof}

Having this partial description of small cuts of $G$, the following is easy to verify.

\begin{lemma} \label{l:feasible}
Each one of the link sets 
$L_{\sf red}=\{tx,ay,yr\}$ and
$L_{\sf blue}=\{tb,rz\}$ 
is an inclusion minimal-feasible solution. 
\end{lemma}
\begin{proof}
We verify that $L_{\sf red}$ covers $\FF^r$ and
that for every red link there is a set in $\FF_r$ that is not covered by any other red link.
\begin{itemize}
\item
$tx$ covers $\{t\}$ and $A$, and it is the only red link that covers $\{t\}$.  
\item
$ay$ covers $X$, and it is the only red link that covers $X$.
\item
$yr$ covers $\FF^r_C \cup \{Y\}$ and it is the only red link that covers $Y$. 
\end{itemize}

We verify that $L_{\sf blue}$ covers $\FF^r $ and
that for every red link there is a set in $\FF^r$ that is not covered by any other blue link.
\begin{itemize}
\item
$tb$ covers each of $\{t\},A,X,Y$ 
and it is the only blue link that covers $\{t\}$.  
\item
$rz$ covers $\FF^r_C$ and it is the only blue link that covers $V \sem \{r\}$.
\end{itemize}

Consequently, 
each one of the link sets $L_{\sf red},L_{\sf blue}$ covers $\FF^r$ (by Lemma~\ref{l:cores}),
and no subset of $L_{\sf red}$ or of $L_{\sf blue}$ covers $\FF^r$. 
\end{proof}

If we run the primal-dual algorithm on the Fig.~\ref{f:eg}(a) instance, 
it can return the red solution $L_{\sf red}$ of cots $5$ 
(and the dual solution $y_S=1$ for every $\FF$-core $S$).
The blue solution $L_{\sf blue}$ has cost $3$. 
Hence we only get ratio $c(L_{\sf red})/c(L_{\sf blue})=5/3$, which is not good enough. 
The idea is to glue $p$ instances in Fig.~\ref{f:eg}(a) 
via the ``axis'' $r,u,s$, such that the blue link $rs$ 
will be ``shared'' by all instances. 
Then the red solution cost will be $5p$,
the blue solution cost will be $p+2$, 
and the ratio between the costs 
will be $\f{5p}{p+2}$, as claimed in Theorem~\ref{t:main}.

\medskip

Formally, the construction is as follows. 
Let $q,p,k$ be integer parameters such that $q \ge 1$, $p \ge 2$, and 
\begin{equation} \label{e:q}
k \ge 2pq+1 
\end{equation}
We take $p$ copies $(G_1,c_1,L_1), \ldots, (G_p,c_p,L_p)$ of the Fig.~\ref{f:g}(a) instance and
obtain a new {\sc Small Cuts Cover} instance $G,c,L$ with $4p+3$ nodes by:
\begin{itemize}
\item
Contracting the nodes $r_1, \ldots, r_p$ of the $p$ instances into a single node $r$, 
contracting $b_1,\ldots,b_p$ into a single node $b$, and 
contracting $z_1,\ldots,z_p$ into a single node $z$.
\item
Replacing the $p(k-q)$ $rb$-edges by $k-pq$ $rb$-edges and 
the $p(k-q-1)$ $bz$-edges by $k-pq-1$ $bz$-edges
\end{itemize}
See Fig.~\ref{f:eg}(a,b) for the cases $p=2$ and $p=4$, respectively.
For higher values of $p$, one can imagine the nodes $r,b,z$ 
as being an axis on which we ''hang'' the $p$ instances.
The axis nodes $r,b,z$ are shared by the glued instances, while the nodes
$a,t,x,y$ have copies $a_i,t_i,x_i,y_i$ in each instance $i=1,\ldots, p$. 
We will use the following notation for the constructed instance:
\begin{itemize}
\item
$G=(V,E)$ is the graph of the instance, $L$ is the set of links,
and $c$ is the cost function.
\item
$L_{\sf red}$ and $L_{\sf blue}$ are 
the red and the blue links of the instance, respectively; 
the links in $\{a_1y_1, \ldots,a_py_p\} \cup \{t_1b, \ldots,t_pb\}$ 
have cost $1$ while all the other links have cost $2$.
\item
$A_i,X_i,Y_i$ are the copies of the cuts $A,X,Y$ of instance $i$, respectively. 
\item
$C=\{z\} \cup (\cup_{i=1}^p \{x_i,y_i\})$ is the set of gray nodes of the instance.
\item
$\FF=\{\empt \ne S \subset V:d(S)<k\}$ 
is the family of small cuts of $G$, $\CC$ is the family of $\FF$-cores.
\item
$\FF^r=\{S \in \FF: r \notin S\}$ is the family of small cuts of $G$ that do not contain $r$.
\item
$\FF^r_C=\{S \in \FF: r \notin S, C \subs S\}$ is the 
family of small cuts of $G$ that contain $C$ but not $r$. 
\end{itemize}

\begin{figure} \centering \includegraphics[scale=0.85]{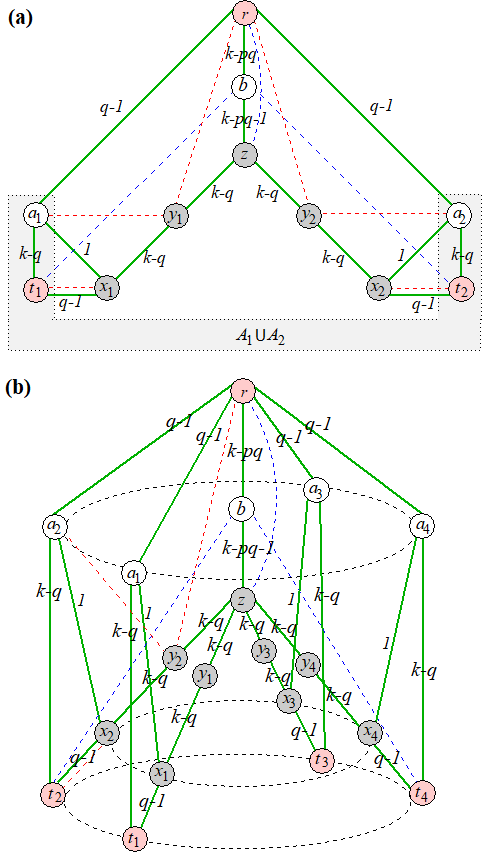}
\caption{
Illustration of the construction. 
The green edges are those of the input graph, and the number on a green edge 
is the multiplicity/capacity of the edge. The red and blue edges are links.
The weight of a link is the number of non-white nodes it connects. 
\begin{enumerate}[(a)]
\item
The case $p=2$; note that the cut $A_1 \cup A_2$ is also a small cut.
\item
The case $p=4$; only some links are shown. 
\end{enumerate}
}
\label{f:eg} \end{figure}

Naturally, $G$ has small cuts inherited from each $G_i$ -- 
these are the cuts $\{t_i\},A_i,X_i,Y_i$ that contain no node from the axis nodes $r,b,z$.
But $G$ has additional small cuts.
For example, any cut that is a union of $1 \le p' \le p$ cuts $A_i=\{t_i,a_i\}$ has 
$p'(2q-1) <k$ edges, and thus by (\ref{e:q}) is small;
see an example of such a cut $A_1 \cup A_2$ in Fig.~\ref{f:eg}(a).
Let $\AA=\left\{\cup_{i \in I} A_i: \empt \ne I \subs \{1, \ldots p\right\}\}$
be the family of such cuts.

\begin{lemma} \label{l:cores'}
The family of $\FF$-cores is
$\CC=\{\{r\},\{t_1\}, \ldots, \{t_p\},C\}$ and
\[
\FF^r \sem \FF^r_C=\{\{t_1\}, \ldots, \{t_p\} \} \cup \AA \cup \{X_1, \ldots,X_p\} \cup \{Y_1, \ldots,Y_p\} \ .
\]
\end{lemma}
\begin{proof}
We prove the first part, that $\CC=\{\{r\},\{t_1\}, \ldots, \{t_p\},C\}$.
Note that:
\begin{itemize}
\item
Each $\{t_i\}$ is a small cut since $d(t_i)=k-1$.
\item
$\{r\}$ is a small cut since $d(r)=k-pq+p(q-1)=k-p$.
\item
$C$ is a small cut since $d(C)=(k-pq-1) +pq=k-1$. 
\end{itemize}
Thus each of the singleton sets $\{r\}, \{t_1\}, \ldots, \{t_p\}$ is a core,
and $C \in \FF$. We prove that $C$ is an inclusion-minimal small cut. 
Suppose to the contrary that there is a small cut $S \subset C$. 
If $z \notin S$ then there is an index $i$ such that 
$S \cap \{x_i,y_i\} \ne \empt$,
but $d(S) \ge d(S \cap \{x_i,y_i\}) \ge k$ by Lemma~\ref{l:cores}. 
If $z \in S$ then $d(S) \ge (k-pq-1)+(k-q) \ge k$ by (\ref{e:q}).
Thus such $S$ does not exists, and $C$ is a core. 

We prove the second part.
We already know that all cuts in 
$\{\{t_1\}, \ldots ,\{t_p\}\} \cup \AA \cup  \{X_1, \ldots, X_p\} \cup \{Y_1, \ldots, Y_p\}$
are in $\FF^r \sem \FF^r_C$; they are all small and none of them contains $C$.
Suppose to the contrary that there is some other small cut $S \in \FF^r \sem \FF^r_C$. 
We now consider two cases and in each case arrive at a contradiction. 

{\bf Case (i): $z \in S$.} 
Note that then $\de(S)$ contains 
$k-pq-1$ $zb$-edges (if $b \notin S$) or  
$k-pq$ $br$ edges (if $b \in S$).
If $y_i \notin S$ for some $i$ then $\de(S)$ also contains $k-q$ $zy_i$-edges 
and we get the contradiction
$d(S) \ge (k-pq-1)+(k-q) \ge k$. 
So assume that $y_i \in S$ for all $i$. 
If $x_i \notin S$ for some $i$ then we get again that 
$d(S) \ge (k-pq-1)+(k-q) \ge k$, so also $x_i \in S$ for all $i$. 
But then $S$ contains $C$, which is a contradiction.

{\bf Case (ii): $z \notin S$.}
Then also $b \notin S$, as otherwise $d(S) \ge d(b)=2k-2pq-1 \ge k$.
Since there are no edges between $\{t_i,a_i,x_i,y_i\}$ and $\{t_j,a_j,x_j,y_j\}$ for $i \ne j$,
$S$ is either a small cut $S=S_i \subs \{t_i, a_i,x_i,y_i\}$ for some $i$, 
or is a union of such cuts.
By Lemma~\ref{l:cores}, the only small cuts contained in $\{t_i, a_i,x_i,y_i\}$ are 
$\{t_i\},A_i,X_i,Y_i$.
All these cuts are in $\FF^r \sem \FF^r_C$, and moreover, any union of the cuts $A_i$ gives a cut in $\AA$.
We claim that no other combinations are possible.
Suppose that $S$ is a union of some cuts,
and one of them $S_i \in \{\{t_i\},A_i,X_i,Y_i\}$ and some other $S_j \in \{\{t_j\},A_j,X_j,Y_j\}$, where $i \ne j$
(if $i=j$ then one of $S_i,S_j$ contains the other).
Then $d(S) \ge d(S_i)+d(S_j)$. 
However, if $S_i \ne A_i$ then $d(S_i) = k-1$, and since $d(S_j) \ge 1$ we get
$d(S) \ge d(S_i)+d(S_j) \ge k$. 
Thus the only unions that give a small cut are the cuts in $\{A_1, \ldots,A_p\}$, namely $S \in \AA$,
contradicting our assumption that $S \notin \AA$.

In both cases we arrived at a contradiction, thus concluding the proof. 
\end{proof}

\begin{lemma} \label{l:feasible'}
Each one of the link sets $L_{\sf red},L_{\sf blue}$ is an inclusion minimal feasible solution. 
\end{lemma}
\begin{proof}
We will verify that $L_{\sf red}$ covers $\FF^r$ and
that for every red link there is a set in $\FF^r$ that is not covered by any other red link.
\begin{itemize}
\item
$t_ix_i$ covers $\{t_i\}$ and every set $A \in \AA$ that contains $t_i$, 
and it is the only red link that covers $\{t_i\}$.  
Thus the link set $\{t_1x_1, \ldots, t_px_p\}$ covers $\{\{t_1\}, \ldots, \{t_p\}\} \cup \AA$. 
\item
$a_iy_i$ covers $X_i$, and it is the only red link that covers $X_i$.
Thus the link set $\{a_1y_1, \ldots, a_py_p\}$ covers $\{X_1, \ldots,X_p\}$. 
\item
$y_ir$ covers $Y_i$ and $\FF^r_C$, and it is the only red link that covers $Y_i$. 
Consequently, the link set $\{y_1r, \ldots, y_pr\}$ covers $\{Y_1, \ldots,Y_p\} \cup \FF^r_C$. 
\end{itemize}

We will verify $L_{\sf blue}$ covers $\FF^r$ and
that for every blue link there is a set in $\FF^r$ that is not covered by any other blue link.
\begin{itemize}
\item
$t_ib$ covers each of $\{t_i\},X_i,Y_i$ and any set $A \in \AA$ that contains $t_i$, 
and it is the only blue link that covers $\{t_i\}$.  
Consequently, the link set $\{t_1b, \ldots t_pb\}$ covers the set family
\[
\FF^r \sem \FF^r_C=\{\{t_1\}, \ldots ,\{t_p\}\} \cup \AA \cup  \{X_1, \ldots, X_p\} \cup \{Y_1, \ldots, Y_p\} \ .
\]
\item
$rz$ covers $\FF^r_C$ and it is the only blue link that covers $V\sem \{r\}$;
note that $V\sem \{r\}$ is a small cut since $d(r)= k-pq+p(q-1)=k-p<k$. 
\end{itemize}

Consequently, 
each one of the link sets $L_{\sf red},L_{\sf blue}$ covers $\FF^r$ (by Lemma~\ref{l:cores'}),
and no subset of $L_{\sf red}$ or of $L_{\sf blue}$ covers $\FF^r$. 
\end{proof}

If we run the primal-dual algorithm on the constructed instance, 
then in Phase~1 it will raise the dual variables 
of the $p+2$ cores in $\CC=\{\{r\},\{t_1\}, \ldots, \{t_p\},C\}$ to $1$ and will add all links,
but all blue links will be removed during Phase~2 (the reverse delete phase). 
The algorithm will thus return the red solution $L_{\sf red}$ of cost $5p$ 
(with the dual solution $y_S=1$ for every $S \in \CC$).
The blue solution $L_{\sf blue}$ has cost $p+2$. 
The ratio in this case is $5p/(p+2)$.
Moreover, if we assign cost $2+\epsilon$ to the blue link $rz$ and cost $1+\epsilon$ to every other blue link, 
no blue link will be even selected at Phase~1. 

The construction works for any integers $q,p,k$ such that 
$q \ge 1$, $p \ge 2$, and $k \ge 2pq+1$.
The graph of the constructed instance has $4p+3$ nodes. 
Thus Theorem~\ref{t:main} is proved.

Note that the highest possible value of $p$ is 
when $q=1$ and $p=\lfloor (k-1)/2 \rfloor$,
giving a gap of 
$5 \cdot \f{k-1}{k+3}=5(1-\f{4}{k+3})$ for $k$ odd and 
$5 \cdot \f{k-2}{k+2}=5(1-\f{4}{k+2})$ for $k$ even.
 
We note that our example does not prove an integrality gap of $5$, 
and does not exclude that some other algorithm based on the standard LP can 
achieve approximation ratio better than $5$. 
However, the iterative rounding method was already excluded in \cite{SBC}.


\end{document}